\def\UseBibLatex{1}%

\documentclass[10pt]{article}%
\providecommand{\SoCGVer}[1]{}%
\providecommand{\NotSoCGVer}[1]{#1}%
\newcommand{\IfPrinterVer}[2]{#2}%

\NotSoCGVer{\usepackage[cm]{fullpage}}%
\usepackage{amsmath}%
\usepackage{amssymb}%
\NotSoCGVer{\usepackage{euscript}}%
\usepackage{xcolor}%
\usepackage{wasysym}%
\usepackage{multirow}%
\usepackage{graphicx} 

\NotSoCGVer{%
\usepackage[amsmath,thmmarks]{ntheorem}%
\theoremseparator{.}%

\usepackage{titlesec}%
\titlelabel{\thetitle. }%

\titleformat{\paragraph}[runin]
{\normalfont\bfseries}
{\theparagraph}
{1em}
{\addperiod}

\newcommand{\addperiod}[1]{#1.}
}

\usepackage{xcolor}%
\usepackage{mleftright}%
\usepackage{xspace}%
\usepackage{hyperref}%
\usepackage{gensymb}

\newlength{\savedparindent}

\newcommand{\RestoreIndent}{\setlength{\parindent}{\savedparindent}}

\ifx\ShowNotes\undefined%
\newcommand{\Notes}[1]{}
\else
\usepackage{framed}%
\newcommand{\Notes}[1]{%
   \noindent
   \begin{leftbar}%
       \noindent%
       \hspace*{-0.05cm}%
       {\textbf{\Large{\underline{Notes}:}}}
       
       \RestoreIndent  %
       #1
   \end{leftbar}
   \medskip%
}
\fi

\IfPrinterVer{%
   \usepackage{hyperref}%
}{%
   \usepackage{hyperref}%
   \hypersetup{%
      breaklinks,%
      ocgcolorlinks, colorlinks=true,%
      urlcolor=[rgb]{0.25,0.0,0.0},%
      linkcolor=[rgb]{0.5,0.0,0.0},%
      citecolor=[rgb]{0,0.2,0.445},%
      filecolor=[rgb]{0,0,0.4},
      anchorcolor=[rgb]={0.0,0.1,0.2}%
   }
}

\SoCGVer{%
   \theoremstyle{plain}%
   \newtheorem{fact}[theorem]{Fact}
   \newtheorem{invariant}[theorem]{Invariant}
   \newtheorem{question}[theorem]{Question}
   \newtheorem{prop}[theorem]{Proposition}
   \newtheorem{openproblem}[theorem]{Open Problem}

   \theoremstyle{plain}%
   \newtheorem{defn}[theorem]{Definition}
   \newtheorem{problem}[theorem]{Problem}
   \newtheorem{xca}[theorem]{Exercise}
   \newtheorem{exercise_h}[theorem]{Exercise}
   \newtheorem{assumption}[theorem]{Assumption}%

   \newtheorem{proofof}{Proof of\!}%
}%

\NotSoCGVer{%
\theoremseparator{.}%

\theoremstyle{plain}%
\newtheorem{theorem}{Theorem}[section]

\newtheorem{lemma}[theorem]{Lemma}

\theoremstyle{plain}%
\theoremheaderfont{\sf} \theorembodyfont{\upshape}%
\newtheorem*{remark:unnumbered}[theorem]{Remark}%
\newtheorem*{remarks}[theorem]{Remarks}%
\newtheorem{remark}[theorem]{Remark}%

\newtheorem{defn}[theorem]{Definition}

\newcommand{\myqedsymbol}{\rule{2mm}{2mm}}

\theoremheaderfont{\em}%
\theorembodyfont{\upshape}%
\theoremstyle{nonumberplain}%
\theoremseparator{}%
\theoremsymbol{\myqedsymbol}%
\newtheorem{proof}{Proof:}%

}

\SoCGVer{%
   \theoremstyle{plain}%
   \newtheorem{remark:unnumbered}[theorem]{Remark}%
}

\definecolor{blue25emph}{rgb}{0, 0, 11}
\providecommand{\emphic}[2]{%
   \textcolor{blue25emph}{%
      \textbf{\emph{#1}}}%
   \index{#2}}

\providecommand{\emphi}[1]{\emphic{#1}{#1}}

\definecolor{almostblack}{rgb}{0, 0, 0.3}
\providecommand{\emphw}[1]{{\textcolor{almostblack}{\emph{#1}}}}%

\newcommand{\atgen}{\symbol{'100}}
\newcommand{\SarielThanks}[1]{\thanks{Department of Computer Science;
      University of Illinois; 201 N. Goodwin Avenue; Urbana, IL,
      61801, USA; {\tt sariel\atgen{}illinois.edu}; {\tt
         \url{http://sarielhp.org/}.} #1}}

\newcommand{\HLink}[2]{\hyperref[#2]{#1~\ref*{#2}}}
\newcommand{\HLinkSuffix}[3]{\hyperref[#2]{#1\ref*{#2}{#3}}}

\newcommand{\figlab}[1]{\label{fig:#1}}
\newcommand{\figref}[1]{\HLink{Figure}{fig:#1}}

\newcommand{\thmlab}[1]{{\label{theo:#1}}}
\newcommand{\thmref}[1]{\HLink{Theorem}{theo:#1}}

\newcommand{\seclab}[1]{\label{sec:#1}}
\newcommand{\secref}[1]{\HLink{Section}{sec:#1}}

\newcommand{\remlab}[1]{\label{rem:#1}}
\newcommand{\remref}[1]{\HLink{Remark}{rem:#1}}%

\newcommand{\tbllab}[1]{\label{table:#1}}
\newcommand{\tblref}[1]{\HLink{Table}{table:#1}}

\newcommand{\lemlab}[1]{\label{lemma:#1}}
\newcommand{\lemref}[1]{\HLink{Lemma}{lemma:#1}}%

\newcommand{\deflab}[1]{\label{def:#1}}

\providecommand{\eqlab}[1]{}%
\renewcommand{\eqlab}[1]{\label{equation:#1}}

\newcommand{\remove}[1]{}%
\newcommand{\Set}[2]{\left\{ #1 \;\middle\vert\; #2 \right\}}
\newcommand{\pth}[2][\!]{\mleft({#2}\mright)}%
\newcommand{\pbrcx}[1]{\left[ {#1} \right]}%
\newcommand{\Prob}[1]{\mathop{\mathbb{P}}\!\pbrcx{#1}}

\newcommand{\ExChar}{\mathbb{E}}%
\newcommand{\Ex}[2][\!]{\mathop{\ExChar}#1\pbrcx{#2}}

\newcommand{\LS}{\Mh{\mathcal{L}}}%
\newcommand{\DS}{\Mh{\mathcal{D}}}%
\newcommand{\DSA}{\Mh{\mathcal{H}}}%

\newcommand{\ceil}[1]{\left\lceil {#1} \right\rceil}

\newcommand{\cardin}[1]{\left| {#1} \right|}%

\renewcommand{\th}{th\xspace}

\renewcommand{\Re}{\mathbb{R}}%

\usepackage{stmaryrd}%

\DeclareFontFamily{U}{BOONDOX-calo}{\skewchar\font=45 }
\DeclareFontShape{U}{BOONDOX-calo}{m}{n}{
  <-> s*[1.05] BOONDOX-r-calo}{}
\DeclareFontShape{U}{BOONDOX-calo}{b}{n}{
  <-> s*[1.05] BOONDOX-b-calo}{}
\DeclareMathAlphabet{\mathcalb}{U}{BOONDOX-calo}{m}{n}
\SetMathAlphabet{\mathcalb}{bold}{U}{BOONDOX-calo}{b}{n}
\DeclareMathAlphabet{\mathbcalb}{U}{BOONDOX-calo}{b}{n}

\providecommand{\BibLatexMode}[1]{}
\providecommand{\BibTexMode}[1]{#1}

\ifx\UseBibLatex\undefined%
  \renewcommand{\BibLatexMode}[1]{}
  \renewcommand{\BibTexMode}[1]{#1}
\else
  \renewcommand{\BibLatexMode}[1]{#1}
  \renewcommand{\BibTexMode}[1]{}
\fi

\newcommand{\UsePackage}[1]{%
  \IfFileExists{../styles/#1.sty}{%
      \usepackage{../styles/#1}%
   }{%
      \IfFileExists{./styles/#1.sty}{%
         \usepackage{styles/#1}%
      }{%
         \usepackage{#1}%
      }%
   }%
}

\BibLatexMode{%
   \UsePackage{sariel_biblatex}%
   \usepackage[english]{babel}%
   \usepackage{csquotes}
   \renewcommand*{\multicitedelim}{\addcomma\space}%
}

\usepackage[inline]{enumitem}

\newlist{compactenumA}{enumerate}{5}%
\setlist[compactenumA]{topsep=0pt,itemsep=-1ex,partopsep=1ex,parsep=1ex,%
   label=(\Alph*)}%

\newlist{compactenuma}{enumerate}{5}%
\setlist[compactenuma]{topsep=0pt,itemsep=-1ex,partopsep=1ex,parsep=1ex,%
   label=(\alph*)}%

\newlist{compactenumI}{enumerate}{5}%
\setlist[compactenumI]{topsep=0pt,itemsep=-1ex,partopsep=1ex,parsep=1ex,%
   label=(\Roman*)}%

\newlist{compactenumi}{enumerate}{5}%
\setlist[compactenumi]{topsep=0pt,itemsep=-1ex,partopsep=1ex,parsep=1ex,%
   label=(\roman*)}%

\newlist{compactitem}{itemize}{5}%
\setlist[compactitem]{topsep=0pt,itemsep=-1ex,partopsep=1ex,parsep=1ex,%
   label=\ensuremath{\bullet}}%

\numberwithin{figure}{section}%
\numberwithin{table}{section}%
\numberwithin{equation}{section}%

\providecommand{\Mh}[1]{#1}%

\providecommand{\G}{\Mh{G}}%
\renewcommand{\G}{\Mh{G}}%
\newcommand{\GA}{\Mh{H}}%

\newcommand{\cC}{\Mh{\mathsf{c}_3}}%

\newcommand{\CM}{\Mh{\mathcal{C}}}%
\newcommand{\CMB}{\Mh{\mathcal{B}}}%
\newcommand{\CMC}{\Mh{\mathcal{W}}}%

\newcommand{\kopt}{\Mh{\mathcalb{m}^\star}}%
\newcommand{\koptX}[1]{\Mh{\mathcalb{m}}^\Mh{\star}\pth{#1}}%
\newcommand{\koptiX}[1]{\Mh{\mathcalb{m}}^\Mh{\star}_{#1}}%

\newcommand{\MOpt}{\Mh{\mathcal{M}}^\Mh{\star}}%
\newcommand{\MOptX}[1]{\MOpt_{#1}}%
\newcommand{\VOpt}{\VV^\Mh{\star}}

\newcommand{\eps}{{\varepsilon}}%

\newcommand{\VFree}{\Mh{F}}%

\providecommand{\p}{}
\providecommand{\q}{}
\renewcommand{\p}{\Mh{p}}%
\renewcommand{\q}{\Mh{q}}%

\newcommand{\PS}{\Mh{P}}%

\newcommand{\diamX}[1]{\text{diam}\pth{#1}}%
\newcommand{\OSet}{\Mh{\mathsf{U}}}%
\newcommand{\obj}{\Mh{o}}

\newcommand{\etal}{\textit{et~al.}\xspace}

\newcommand{\VV}{\Mh{\mathsf{V}}}%
\newcommand{\VX}[1]{\VV\pth{#1}}%
\newcommand{\EG}{\Mh{\mathsf{E}}}%
\newcommand{\EGX}[1]{\Mh{E}\pth{#1}}

\newcommand{\DFS} {\TermI{DFS}\xspace}

\newcommand{\Term}[1]{\textsf{#1}}
\newcommand{\TermI}[1]{\Term{#1}\index{#1@\Term{#1}}}

\newcommand{\MXor}{\Mh{\mathcal{X}}}%

\newcommand{\densityOp}{\Mh{\mathop{\mathrm{density}}}}%
\newcommand{\densityX}[1]{\densityOp\pth{#1}}%
\newcommand{\cDensity}{\Mh{\rho}} %

\newcommand{\sphereC}{{\mathbb{{S}}}}%

\newcommand{\SetA}{\Mh{X}}%
\newcommand{\SetB}{\Mh{Y}}%
\newcommand{\SetC}{\Mh{U}}

\newcommand{\SepSet}{\Mh{Z}}%

\newcommand{\IGraphC}{\Mh{\mathcal{I}}}%
\newcommand{\IGraphX}[1]{\IGraphC_{#1}}

\newcommand{\GInduced}[1]{\G_{|{#1}}}
\newcommand{\objA}{\Mh{g}}%

\newcommand{\cSep}{\Mh{\zeta}}%
\newcommand{\Entity}{\mathcal{H}}
\newcommand{\wX}[1]{\Mh{\overline{\mathsf{w}}}\pth{#1}}%
\newcommand{\epsA}{\Mh{\xi}}%
\newcommand{\BadProb}{\Mh{\gamma}}%
\newcommand{\Llen}{\Mh{L_{\mathrm{len}}}}%
\newcommand{\disk}{\Mh{\Circle}}%

\newcommand{\listX}[1]{\Mh{\ell}\pth{#1}}%

\newcommand{\cell}{\Box}

\newcommand{\Arr}{\Mh{\EuScript{A}}}
\newcommand{\ArrX}[1]{\Arr\pth{#1}}%

\newcommand{\VCM}{\VV_\CM}
\newcommand{\nCM}{\Mh{n}_\CM}
\newcommand{\ranges}{\Mh{\mathcal{R}}}%

\newcommand{\spread}{\Mh{\Phi}}

\newcommand{\cEps}{c_\eps}%
\newcommand{\AGC}{\Mh{\mathcal{B}}}%
\providecommand{\TPDF}[2]{\texorpdfstring{#1}{#2}}

\newlength{\arxivwidth}

\newcommand{\arXiv}{%
   \begin{minipage}{\arxivwidth}
       \vspace*{0.1\fontcharht\font`A}%
       \includegraphics[width=\arxivwidth]{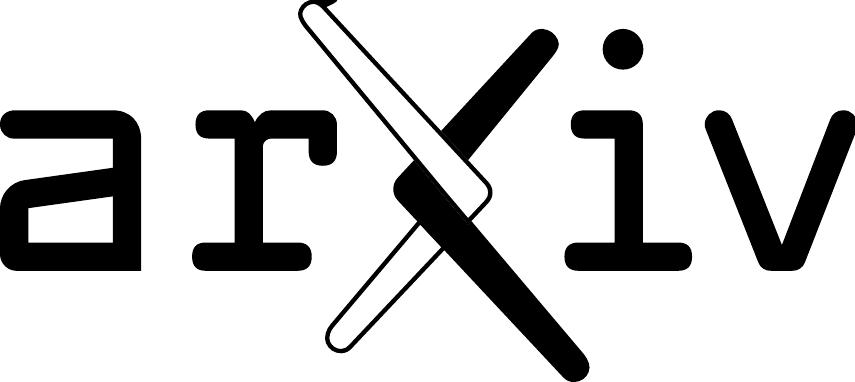}
   \end{minipage}\xspace
}

\BibLatexMode{\bibliography{matchings}}%

\title{Approximation Algorithms for Maximum Matchings in Geometric
   Intersection Graphs}%

\NotSoCGVer{%
   \author{%
   Sariel Har-Peled\SarielThanks{}%
   \and%
   Everett Yang}
}%

\SoCGVer{%
   \author{Sariel Har-Peled}%
   {Department of Computer Science, University of Illinois, 201
      N. Goodwin Avenue, Urbana, IL 61801, USA}%
   {sariel@illinois.edu}%
   {https://orcid.org/0000-0003-2638-9635}%
   {Work on this paper was partially supported by a NSF AF award
      CCF-1907400.}%
   \author{Everett Yang}%
   {Department of Computer Science, University of Illinois, 201
      N. Goodwin Avenue, Urbana, IL 61801, USA}%
   {esyang@illinois.edu}%
   {https://orcid.org/}%
   {}%
}%

\SoCGVer{%
   \authorrunning{S. Har-Peled and E. Yang} %
   \Copyright{Sariel Har-Peled and Everett Yang}%
   \ccsdesc[500]{Theory of computation~Computational geometry}%
   \keywords{Matchings, disk intersection graphs, approximation
      algorithms }%
}

\begin{document}
\maketitle

\begin{abstract}
    We present a $(1- \eps)$-approximation algorithms for maximum
    cardinality matchings in disk intersection graphs -- all with near
    linear running time. We also present an estimation algorithm that
    returns $(1\pm \eps)$-approximation to the size of such matchings
    -- this algorithm runs in linear time for unit disks, and
    $O(n \log n)$ for general disks (as long as the density is
    relatively small).
\end{abstract}

\section{Introduction}

\paragraph*{Geometric intersection graphs}
Given a set of $n$ objects $\OSet$, its intersection graph,
$\IGraphX{\OSet}$, is the graph where the vertices correspond to
objects in $\OSet$ and there is an edge between two vertices if their
corresponding objects intersect. Such graphs can be dense (i.e., have
$\Theta(n^2)$ edges), but they have a linear size representation. It
is natural to ask if one can solve problems on such graphs more
efficiently than explicitly represented graphs.

\paragraph*{Maximum matchings}
Computing maximum cardinality matchings is one of the classical
problems on graphs (surprisingly, the algorithm to solve the bipartite
case goes back to work by Jacobi in the mid 19\th century). The
fastest combinatorial algorithm (ignoring polylog factors) seems to be
the work by Gabow and Tarjan \cite{gt-fsagg-91}, running in
$O( m \sqrt{n})$ time where $m$ is the number of edges in the graph.
Harvey \cite{h-aamm-09} and Mucha and Sankowski \cite{ms-mmge-04}
provided algorithms based on algebraic approach that runs in
$O(n^{\omega})$ time, where $O(n^{\omega})$ is the fastest time known
for multiplying two $n\times n$ matrices. Currently, the fastest known
algorithm for matrix multiplication has $\omega \approx 2.3728596$,
but it is far from being practical.

\paragraph*{Matchings for planar graphs and disk intersection graphs}
Mucha and Sankowski \cite{ms-mmpgg-06} adapted their algebraic
technique for planar graphs (specifically using separators), getting
running time $O(n^{\omega/2})\approx O(n^{1.17})$.  Yuster and Zwick
\cite{yz-mmge-07} adapted this algorithm for graphs with excluded
minors.

\paragraph*{Maximum matchings in geometric intersection graphs}

Bonnet \etal \cite{bcm-mmgig-20} studied the problem for geometric
intersection graphs. For simplicity of exposition, we describe their
results in the context of disk intersection graphs. Given a set $n$
disks with maximum density $\cDensity$ (i.e., roughly the maximum
number of disks covering a point in the plane), they presented an
algorithm for computing maximum matchings with running time
$O( \cDensity^{3\omega/2} n^{\omega/2}) \approx O(\cDensity^{3.5}
n^{1.17})$. This compares favorably with the naive algorithm of just
plugging such graphs into the algorithm of Gabow and Tarjan, which
yields running time $O( m\sqrt{n} ) = O( \cDensity n^{3/2})$.  If the
ratio between the smallest disk and largest disk is at most $\spread$,
they presented an algorithm with running time
$O(\spread^{12\omega}n^{\omega/2})$. Note that the running time of all
these algorithms is super linear in $n$.

\paragraph*{Approximate maximum matchings}

It is well known that it is enough to augment along paths of length up
to $O(1/\eps)$ if one wants $(1-\eps)$-approximate matchings. For
bipartite graphs this implies that one need to run $O(1/\eps)$ rounds
of paths finding stage of the bipartite matching algorithm of Hopcroft
and Karp \cite{hk-nammb-73}. Since such a round takes $O(m)$ time,
this readily leads to an $(1-\eps)$-approximate bipartite matching
algorithm in this case. The non-bipartite case is significantly more
complicated, and the weighted case is even more
difficult. Nevertheless, Duan and Pettie \cite{dp-ltamw-14} presented
an algorithm with running time $O(m\eps^{-1} \log{\eps^{-1}})$ which
provides $(1-\eps)$-approximation to the maximum weight matching in
non-bipartite graph.

\paragraph*{Density and approximate matchings}

For a set of objects $\OSet$ in $\Re^d$, the density of an object is
the number of bigger objects in the set intersecting it. The
\emphw{density} of the set of objects is the maximum density of the
objects. The density is denoted by $\cDensity$, and the premise is
that for real world inputs it would be small. The intersection graph
of such objects when $\cDensity$ is a constant are known as \emphw{low
   density graphs}, have some nice properties, such as having
separators. See \cite{hq-aapel-17} and references therein. In
particular, for a set of fat objects, the density and the maximum
depth (i.e., the maximum number of object covering any point) are
roughly the same. It is well known that low density graphs are sparse
and have $O( \cDensity n)$ edges, where $n$ is the number of objects
in the set. Since one can compute the intersection graph in
$O(n \log n + \cDensity n)$ time, and plug it into the algorithm of
Duan and Pettie \cite{dp-ltamw-14}, it follows that an approximation
algorithm with running time
$O\bigl(n \log n + (\cDensity n / \eps) \log (1/\eps) \bigr)$. See
\secref{sparse} for details. Thus, the challenge is to get better
running times than this baseline.

\subsection{Our results}

Our purpose here is to develop near linear time algorithms for
approximate matchings for the unit disk graph and the general disk
graph cases. Our results are summarized in \tblref{results}.  Note, in
this paper, we assume the input to our algorithms to be a set of
disks.

\begin{compactenumA}
    \item \textsf{Unit disk graph}.

    \begin{compactenumI}
        \item \textsf{Greedy matching.}  We show in \secref{u:greedy}
        a linear time algorithm for the case of unit disks graph --
        this readily provides a $1/2$-approximation to the maximum
        matching. The algorithm uses a simple grid to ``capture'
        intersections, and then use the locality of the grid to find
        intersection with the remaining set of disks.

        \medskip%
        \item \textsf{$(1-\eps)$-approximation.}  In
        \secref{u:bounded:spread} we show how to get a
        $(1-\eps)$-approximation. The running time can be bounded by
        $O\bigl((n/\eps^2) \log (1/\eps)\bigr)$. If the diameter of
        union of disks is at most $\Delta$, then the running time is
        $O\bigl(n + (\Delta^2/\eps^2) \log (1/\eps)\bigr)$.

        \medskip%
        \item \textsf{$(1-\eps)$-estimation.}  Surprisingly, one can
        do even better -- we show in \secref{u:estimate} how to use
        importance sampling to get $(1\pm\eps)$-approximation (in
        expectation) to the size of the maximum matching in time
        $O(n + \textrm{poly}( \log n, 1/\eps))$.
    \end{compactenumI}

    \medskip%
    \item \textsf{Disk graph}.  The general disk graph case is more
    challenging.

    \begin{compactenumI}
        \item \textsf{Greedy matching.}  The greedy matching algorithm
        can be implemented in $O(n \log n)$ time using sweeping, see
        \lemref{greedy:g} (this algorithm works for any nicely behaved
        shapes).

        \medskip%
        \item \textsf{Approximate bipartite case.}  Here, we are given
        two sets of disks, and consider only intersections across the
        sets as edges.  This case can be solved using range searching
        data-structures as was done by Efrat \etal \cite{eik-ghbmr-01}
        -- they showed how to implement a round of the bipartite
        matching algorithm of Hopcroft and Karp \cite{hk-nammb-73}
        using $O(n)$ queries. It is folklore that running $O(1/\eps)$
        rounds of this algorithm leads to a $(1-\eps)$-approximation
        algorithm.  Coupling this with the data-structure of Kaplan
        \etal \cite{kmrss-dpvdg-20} readily leads to a near linear
        approximation algorithm in this case, see
        \secref{g:bipartite}.

        \medskip%
        \item \textsf{$(1-\eps)$-approximation algorithm.}
        Surprisingly, approximate general matchings can be reduced to
        bipartite matchings via random coloring. Specifically, one can
        compute $(1-\eps)$-approximate matchings using
        $2^{O(1/\eps)} \log n$ invocations of approximate bipartite
        matchings algorithm mentioned above. This rather neat idea is
        due to Lotker \etal \cite{lpp-idam-15} who used in the context
        of parallel matching algorithms. This leads to a near linear
        time algorithm for $(1-\eps)$-approximate matchings for disk
        intersection graphs, see \secref{a:g:matching:r:c} for
        details. The running time of the resulting algorithm is
        $2^{O(1/\eps)} n \log^{O(1)} n$.

        We emphasize that this algorithm assumes nothing about the
        density of the input disks.

        \medskip%
        \item \textsf{$(1-\eps)$-estimation.}  One can get an
        $O( n \log n)$ time estimation algorithm in this case, but one
        needs to assume that the input disks have ``small'' density
        $\cDensity$. Specifically, one computes a separator hierarchy
        and then use importance sampling on the patches, as to
        estimate the sampling size. The details require some care, see
        \secref{recursiveSeparatorApproach}. The resulting algorithm
        has running time $O(n \log n)$ if the density is $o(n^{1/9})$.
        
    \end{compactenumI}

    \item \textsf{General shapes.} Somewhat surprisingly, almost all
    our results extends in a verbatim fashion to intersection graphs
    of general shapes. We need some standard assumptions about the
    shapes:
    \begin{compactenumi}
        \item the boundary of any pair of shapes intersects only a
        constant number of times,

        \item these intersections can be computed in constant time,

        \item one can compute the $x$-extreme and $y$-extreme points
        in a shape in constant time,

        \item one can decide in constant time if a point is inside a
        shape, and 

        \item the boundary of a shape intersects any line a constant
        number of times, and they can be computed in constant time.
    \end{compactenumi}
    For fat shapes of similar size, we also assume the diameters of
    all the shapes are the same up to a constant factor, and any
    object $\obj$
    contains a disk of radius $\Omega( \diamX{\obj})$.

    Since the modification needed to make the algorithms work for the
    more general cases are straightforward, we describe the algorithms
    for disks.
\end{compactenumA}

The full version of the paper is available on the \arXiv
\cite{hy-aammg-22}%
\SoCGVer{ and includes all missing details}.

\begin{table}
    \centering
    \begin{tabular}{|c*{4}{|l}|}
      \hline
      Shape
      & Quality
      & Running time
      & Ref
      & Comment
      \\
      \hline
      \multirow{4}{*}{%
      \begin{minipage}{2.3cm}
          \begin{center}
              Unit disks\\
              Fat shapes of\\
              similar size
      \end{center}
      \end{minipage}
      }
      &%
        $1/2$%
      &%
        $O(n)$%
      &%
        \lemref{unit:disk:g}%
      &%
        Greedy
      \\
      \cline{2-5}
      &
        \multirow{2}{*}{$1-\eps$}
      &        
        $O\bigl( \frac{n}{\eps} \log \frac{1}{\eps}\bigr)\Bigr.$
      &
        \lemref{m:b:diam:unit:disks:2}%
      &
      \\        
      \cline{3-5}
      &%
      &
        $O\bigl(n + (\Delta^2/{\eps^{2}}) \log \frac{1}{\eps} \bigr)\Bigr.$
      &
        \remref{bounded:diam}%
      &%
        $\Delta$ = Diam. disks
      \\
      \cline{2-5}
      &
        $1\pm\eps$%
      &        
        $O(n + \eps^{-6} \log^2 n )\Bigr.$
      &
        \thmref{u:estimate}%
      &
        Estimation
      \\
      \hline
      Unit disks
      & 
        $1-\eps$%
      &
        $O\bigl( \frac{n}{\eps} \log \frac{1}{\eps} \bigr)\Bigr.$
      &
        \lemref{m:b:diam:unit:disks:2}%
      &
        Approximate
      \\
      \hline
      \hline
      \multirow{2}{*}{%
      Disks
      }
      &%
        $1-\eps$
      &
        $O\bigl((n/\eps) \log^{11} n \bigr)\Bigr.$
      &
        \lemref{bipartite} ($\star$)
      &
        Bipartite
      \\
      \cline{2-5}
      &%
        $1-\eps$
      &
        $O(2^{O(1/\eps)} n \log^{12} n)\Bigr. $
      &
        \thmref{a:matching:c:c} ($\star$)%
      &
        General
      \\
      \hline
      \hline
      \multirow{6}{*}{%
      Shapes
      }
      &
        $1/2$
      &
        $O(n \log n)$
      &
        \lemref{greedy:g}%
      &
        Greedy
      \\
      \cline{2-5}
      &%
        $1-\eps$
      &
        $O( n \log n +\frac{m}{\eps} \log \frac{1}{\eps}) \Bigr.$
      &
        \lemref{plug:and:play}
      &
        $m:$ \# edges
      \\
      \cline{2-5}
      &%
        $1-\eps$
      &
        $O( n \log n + \frac{n \cDensity}{\eps} \log \frac{1}{\eps})\Bigr.$
      &
        \lemref{low:density}
      &
        $\cDensity$: Density
      \\
      \cline{2-5}
      &%
        Exact
      &
        $O(n \log n)$
      &
        \lemref{small:matching}%
      &
        \begin{minipage}{2cm}
            ~
            
            Matching size\\
            $O( n^{1/8})\Bigr.$
        \end{minipage}
      \\
      \cline{2-5}
      &%
        $1\pm \eps$
      &
        $O(n \log n + \cDensity^{9}\eps^{-19} \log^2 n) \Bigr.$
      &
        \thmref{g:d::estimte}%
      &
        Estimation
      \\
      \hline
    \end{tabular}
    \caption{Results. The ($\star$) indicates the result works only
       for disks.}
    \tbllab{results}
\end{table}

\section{Preliminaries}

\subsection{Notations}

For a graph $\G$, let $\MOptX{\G}$ denote the maximum cardinality
matching in $\G$. Its size is denoted by
$\kopt = \koptX{\G} = \cardin{\MOptX{\G}}$. For a graph
$\G = (\VV,\EG)$, and a set $X \subseteq \VV$ the \emphw{induced
   subgraph} of $\G$ over $X$ is
$\GInduced{X} = \bigl( X, \Set{ uv \in \EG}{ u, v \in X}\bigr)$.  For
a set $Z$, let $\G - Z$ denote the graph resulting from $\G$ after
deleting from it all the vertices of $Z$. Formally, $\G - Z$ is the
graph $\GInduced{\VV \setminus Z}$.

\begin{defn}
    For a set of objects $\OSet$, the \emphi{intersection graph} of
    $\OSet$, denoted by $\IGraphX{\OSet}$, is the graph having $\OSet$
    as its set of vertices, and there is an edge between two objects
    $\obj,\objA \in \OSet$ if they intersect.  Formally,
    \begin{equation*}
        \IGraphX{\OSet}%
        =%
        \pth{\bigl.\OSet, %
           \Set{\bigl.\obj \objA}{\obj, \objA \in \OSet %
              \text{ and } \obj \cap \objA \neq \emptyset}}.        
    \end{equation*}
\end{defn}

For a point $\p \in \Re^2$, and a set of disks $\DS$, let 
\begin{equation*}
    \DS \sqcap \p = \Set{ \disk \in \DS}{ \p \in \disk}
\end{equation*}
be the set of disks of $\DS$ that contain $\p$. Note, that the
intersection graph $\IGraphX{ \p} = \IGraphX{\DS \cap \p}$ is a
clique.

\begin{defn}
    Consider a set of disks $\DS$. A set $\DS' \subseteq \DS$ is an
    \emphi{independent set} (or simply \emphw{independent}) if no pair
    of disks of $\DS'$ intersects.
\end{defn}

\subsection{Low density and separators}

The following is standard by now, see Har-Peled and Quanrud
\cite{hq-aapel-17} and references therein.

\begin{defn}%
  \deflab{low:density}%
  A set of objects $\OSet$ in $\Re^d$ (not necessarily convex or
  connected) has \emphi{density $\cDensity$} if any object $\obj$ (not
  necessarily in $\OSet$) intersects at most $\cDensity$ objects in
  $\OSet$ with diameter equal or larger than the diameter of
  $\obj$. The minimum such quantity is denoted by $\densityX{\OSet}$.
  A graph that can be realized as the intersection graph of a set of
  objects $\OSet$ in $\Re^d$ with density $\cDensity$ is
  \emphi{$\cDensity$-dense}.  The set $\OSet$ is \emphi{low density}
  if $\cDensity = O(1)$. %
\end{defn}%

\begin{defn}%
    \deflab{separator}%
    Let $\G = (\VV,\EG)$ be an undirected graph.  Two sets
    $\SetA, \SetB \subseteq \VV$ are \emphi{separate} in $\G$ if
    \begin{compactenumi}
        \smallskip%
        \item $\SetA$ and $\SetB$ are disjoint, and
        \smallskip%
        \item there is no edge between the vertices of $\SetA$ and the
        vertices of $\SetB$ in $\G$.
    \end{compactenumi}
    \smallskip%
    For a constant $\cSep \in (0,1)$, a set $\SepSet \subseteq \VV$ is
    a \emphi{$\cSep$-separator} for a set $\SetC \subseteq \VV$, if
    $\SetC \setminus \SepSet$ can be partitioned into two
    \emph{separate} sets $\SetA$ and $\SetB$, with
    \begin{math}
        \cardin{\SetA} \leq \cSep\cardin{\SetC}
    \end{math}
    and
    \begin{math}
        \cardin{\SetB} \leq \cSep \cardin{\SetC}.
    \end{math}
\end{defn}

\begin{lemma}[\cite{hq-aapel-17}]
    \lemlab{w:s:low:density}%
    Let $\OSet$ be a set of $n$ objects in $\Re^d$ with density
    $\cDensity$.  One can compute, in expected linear time, a sphere
    $\sphereC$ that intersects in expectation
    $\tau = O\pth{\cDensity + \cDensity^{1/d} n^{1-1/d} }$ objects of
    $\OSet$. The sphere is computed by picking uniformly its radius
    from some range of the form $[\alpha, 2\alpha]$.  Furthermore, the
    total number of objects of $\OSet$ strictly inside/outside
    $\sphereC$ is at most $\cSep n$, where $\cSep$ is a constant that
    depends only on $d$.  Namely, the intersection graph
    $\IGraphX{\OSet}$ has a separator of size $\tau$ formed by all the
    objects of $\OSet$ intersecting $\sphereC$.
\end{lemma}

\subsection{Importance sampling}

Importance sampling is a standard technique for estimating a sum of
terms. Assume that for each term in the summation, one can quickly get
a coarse estimate of its value. Furthermore, assume that better
estimates are possible but expensive. Importance sampling shows how to
sample terms in the summation, then acquire a better estimate {\em
   only for the sampled terms}, to get a good estimate for the full
summation. In particular, the number of samples is bounded
independently of the original number of terms, depending instead on
the coarseness of the initial estimates, the probability of success,
and the quality of the final output estimate.

\begin{lemma}[\cite{bhrrs-eeiso-20}]
    \lemlab{importance:alg}%
    Let $(\Entity_1, w_1,e_1), \ldots, (\Entity_r,w_r,e_r)$ be given,
    where $\Entity_i$'s are some structures, and $w_i$ and $e_i$ are
    numbers, for $i=1, \ldots, r$. Every structure $\Entity_i$ has an
    associated weight $\wX{\Entity_i} \geq 0$ (the exact value of
    $\wX{\Entity_i}$ is not given to us). In addition, let
    $\epsA > 0$, $\BadProb$, $b$, and $M$ be parameters, such that:
    \smallskip%
    \begin{compactenumi}[leftmargin=3em]
        \item $\forall i \quad w_i,e_i \geq 1$, \smallskip%
        \item $\forall i \quad e_i/b \leq \wX{\Entity_i} \leq e_i b$,
        and \smallskip%
        \item $\Gamma = \sum_i w_i \cdot \wX{\Entity_i} \leq M$.
    \end{compactenumi}
    \smallskip%
    Then, one can compute a new sequence of triples
    $(\Entity_1', w_1', e_1'), \ldots, (\Entity_t',w_t', e_t')$, that
    also complies with the above conditions, such that the estimate
    \begin{math}
        Y = \sum_{i=1}^t w_i' \wX{\Entity_i'}
    \end{math}
    is a multiplicative $(1\pm \epsA)$-approximation to $\Gamma$, with
    probability $\geq 1- \BadProb$.  The running time of the algorithm
    is $O( r)$, and size of the output sequence is
    \begin{math}
        t = O\pth{ b^4 \epsA^{-2} (\log \log M + \log \BadProb^{-1})
           \log M }.
    \end{math}
\end{lemma}

\begin{remark}
    \remlab{L:2}%
    The algorithm of \lemref{importance:alg} does not use the
    entities $\Entity_i$ directly at all. In particular, the
    $\Entity_i'$s are just (reweighed) copies of some original
    structures. The only thing that the above lemma uses is the
    estimates $e_1, \ldots, e_r$ and the weights $w_1, \ldots, w_r$.

    (B) We are going to use \lemref{importance:alg}, with
    $\epsA = O(\eps)$, $\BadProb = 1/n^{O(1)}$, $b = 2$, and $M =
    n$. As such, the size of the output list is 
    \begin{math}
        \Llen%
        =%
        O(\eps^{-2} \log^2 n )
    \end{math}
\end{remark}

\subsection{Background on matchings}

For a graph $\G$, a \emphi{matching} is a set $\CM \subseteq \EGX{\G}$
of edges, such that no pair of them not share an endpoint. A matching
that has the largest cardinality possible for a graph $\G$, is a
\emphi{maximum} matching.  Given a graph $\G$ and a matching $\CM$ on
$\G$, an \emphi{alternating path} is a path with edges that alternate
between matched edges (i.e., edges that are in $\CM$) and unmatched
edges (i.e., edges in $\EGX{\G} \setminus \CM$). If both endpoints of
an alternating path are unmatched (i.e., \emphi{free}), then it is an
\emphi{augmenting path}. In the following, let $\MOpt$ denote a
maximum cardinality matching in $\G$, and let
$\kopt = \koptX{\G} = |\MOpt|$ denote its size. For $\beta \in [0,1]$,
a matching $\CMB \subseteq \EGX{\G}$ is an \emphi{$\beta$-matching}
(or $\beta$-approximate matching) if $\cardin{\CMB} \geq \beta \kopt$.
The set of vertices covered by the matching $\CMB$ is denoted by
$\VX{\CMB} = \displaystyle\bigcup_{uv \in \CMB} \{ u, v\}$.

The \emphw{length} of a path is the number of its edges.

\SoCGVer{The following claim is well known.}

\begin{lemma}
    \lemlab{augmenting}%
    For any $\eps \in (0,1)$, if $|\CM| < (1 - \eps)|\MOpt|$, then
    there are at least $(\eps/2) |\MOpt|$ disjoint augmenting paths of
    $\CM$, each of length at most $4 / \eps$.
\end{lemma}
\NotSoCGVer{%
\begin{proof}
    This is well known, and we include a proof for the sake of
    completeness.  Suppose we are given the current matching $\CM$ and
    the optimal matching $O$.  Let $k = |O|$.  Consider the symmetric
    difference
    $\MXor = \CM \oplus \MOpt = (\CM \setminus \MOpt) \cup (\MOpt
    \setminus \CM)$ -- it is a collection of alternating paths and
    cycles.  For a path $\pi \in \MXor$, its \emphw{contribution} is
    $\beta(\pi) = |\MOpt \cap \pi| - |\CM \cap \pi| \in \{-1, 0,
    +1\}$.  Let $\Pi$ be the set of all the augmenting paths in
    $\MXor$ (an augmenting path has contribution of $+1$).  Observe
    that
    \begin{equation}
        |\Pi|
        \geq%
        \sum_{\pi \in \MXor} \beta(\pi)%
        \geq%
        |\MOpt| - |\CM| > |\MOpt| - (1 - \eps)|\MOpt| = \eps k.
    \end{equation}
    We have that $|\MXor| \leq |\CM| + |\MOpt| \leq 2k$, and as such,
    the average length of an augmenting path in $\Pi$ is at most
    $2k/ (\eps k ) = 2/\eps$. By Markov's inequality, at most half of
    them can be twice larger than the average, which implies the
    claim.
\end{proof}
}

\section{Approximate matchings for unit disk graph}

\subsection{Greedy maximal matching}
\seclab{u:greedy}

In a graph $\G$, the greedy maximal matching can be computed by
repeatedly picking an edge of $\G$, adding it to the matching,and
removing the two vertices of the edges from $\G$. We do this
repeatedly until no edges remain. The resulting \emphi{greedy
   matching} is a maximal matching, and every maximal matching is a
$1/2$-approximation to the maximum matching. To avoid the
maximum/maximal confusion, we refer to such a matching as a greedy
matching.

\begin{lemma}
    \lemlab{unit:disk:g}%
    Let $\DS$ be a set of $n$ unit disks in the plane, where a
    \emphw{unit disk} has radius one. One can compute, in $O(n)$ time,
    a $(1/2)$-approximate matching for $\IGraphX{\DS}$, where
    $\IGraphX{\DS}$ is the intersection graph of the disks of $\DS$.
\end{lemma}

\begin{proof}
    For every disk, compute all the integral grid points that it
    covers. Every disk covers at least one, and at most five grid
    points. We use hashing to compute for every grid point the disks
    that covers it. These lists can be computed in $O(n)$ time
    overall. Next, for every  grid point that stores more than one
    disk,  scan it, and break it into pairs, where every pair is
    reported as a matching edge, and the two disks involved are
    removed.

    By the end of this process, we computed a partial matching $\CM$,
    and we have a set $\DS'$ of leftover disks that are not matched
    yet. The disks of $\DS'$ cover every integral grid point at most
    once. Using the hash table one can look for intersections -- for
    every grid point that is active (i.e., has one disk of $\DS'$
    covering it), the algorithm lookup in the hash table any disk that
    covers any of the $8$ neighboring grid points. Each such
    neighboring point offers one disk that might intersect the current
    disk. If we find an intersecting pair, the algorithm outputs it
    (removing the two disks involved). This requires $O(1)$ time per
    active grid point, and linear time overall. At the end of this
    process, all the remaining disks are disjoint, implying that the
    computed matching is maximal and thus a $(1/2)$-approximation to
    the maximum matching.
\end{proof}

\subsection{\TPDF{$(1-\eps)$}{1-epsilon}-approximation}
\seclab{u:bounded:spread}

\begin{lemma}
    \lemlab{m:b:diam:unit:disks}%
    Let $\DS$ be a set of $n$ unit disks, and let $\eps \in (0,1)$ be
    a parameter. Then, one can compute, in
    $O\bigl((n/\eps^2) \log (1/\eps)\bigr)$ time, an
    $(1-\eps)$-matching in $\IGraphX{\DS}$, where $\IGraphX{\DS}$ is
    the intersection graph of $\DS$.

    If the diameter of $\cup \DS$ is $\Delta$, then the running time
    is $O\bigl(n + (\Delta^2/\eps^3) \log (1/\eps)\bigr)$.
\end{lemma}

\newcommand{\nT}{\Mh{\tau}}

\begin{proof}
    Using a unit grid, the algorithm computes for each disk in $\DS$ a
    grid point that it contains, and register the disk with this
    point. Using hashing this can be done in $O(n)$ time overall. For
    a grid point $\p$, let $\listX{\p}$ be the list of disks that are
    registered with it.  Let $\p_1, \ldots \p_\nT$ be the points with
    non-empty lists.
    
    For a point $\p_i$, the graph $\IGraphX{\DS \cap \p_i}$ is the
    \emphi{tower} of $\p_i$. Consider a maximum matching $\MOpt$ of
    $\IGraphX{\DS}$. An edge $uv \in \MOpt$ is a \emphi{cross edge} if
    $u$ and $v$ belong to two different towers. Observe that if there
    are two cross edges between two towers in a matching, then we can
    exchange them by two edges internals to the two towers, preserving
    the size of the matching -- this observation is due to Bonnet
    \etal \cite{bcm-mmgig-20}. As such, we can assume that there is at
    most one cross edge between any two towers in the maximum matching
    $\MOpt$. In addition, any tower can have edges only with towers in
    its neighborhood -- specifically, two towers might have an edge
    between them, if the distance between their centers is at most
    $4$.  As such, the number of cross edges in $\MOpt$ is at most
    $24\nT = 48\nT/2$, as each tower interacts with at most $48$ other
    towers, see \figref{towers}.

    \begin{figure}
        \phantom{}%
        \hfill%
        \includegraphics[page=1]{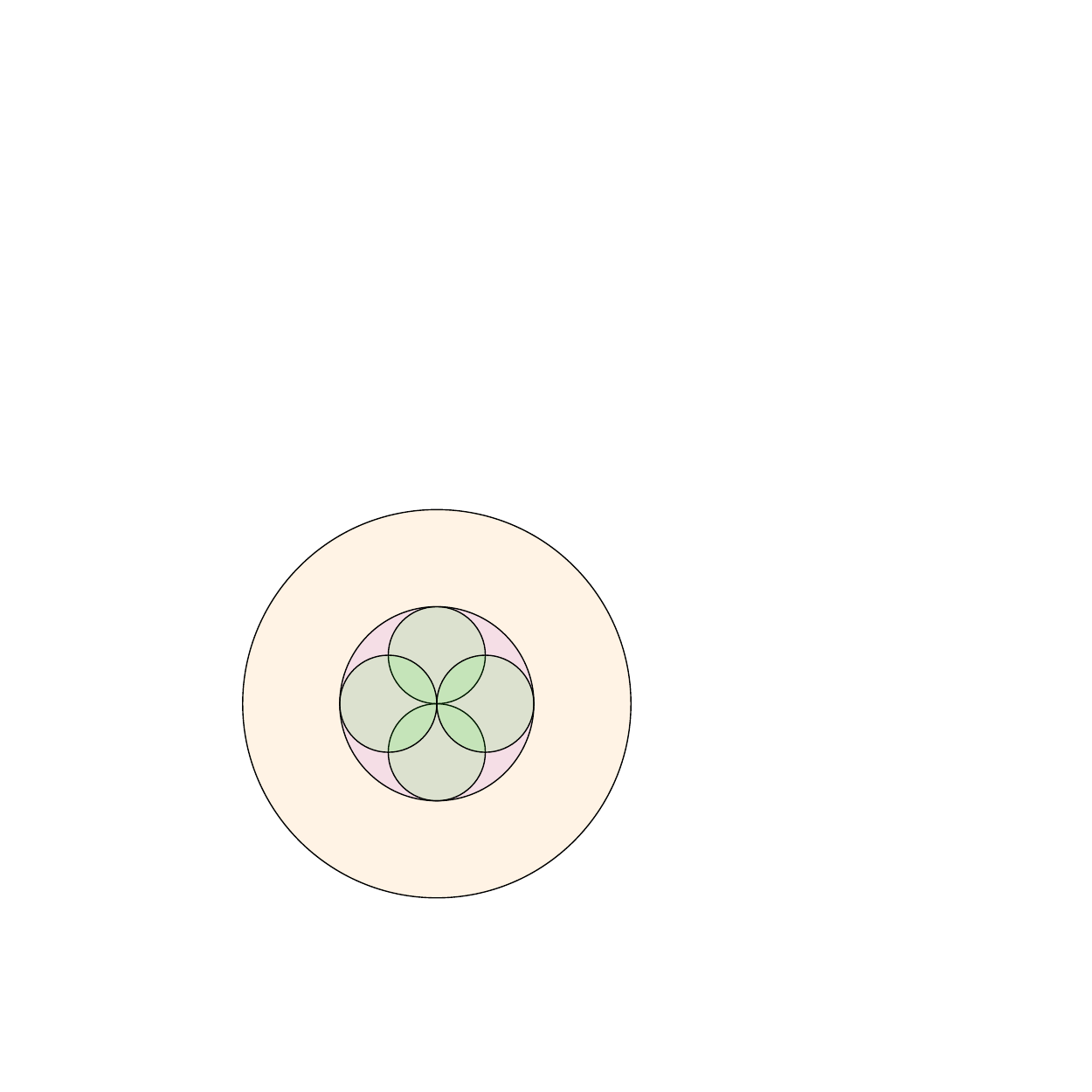}%
        \hfill%
        \includegraphics[page=2]{figs/tower}%
        \hfill%
        \phantom{}%
        \caption{A tower can interact with at most 48 other towers.}
        \figlab{towers}
    \end{figure}

    If a tower has more than (say) $200/\eps$ disks in it, then we add
    the greedy matching in the tower to the output, and remove all the
    disks in the tower. This yields a matching of size at least
    $100/\eps$, that destroys at most $48$ additional edges from the
    optimal matching. Thus, it is sufficient to compute the
    approximate matching in the residual graph. We repeat this process
    till all the remaining towers have at most $O(1/\eps)$ disks in
    them. Let $\DS'$ be the remaining set of disks -- by the bounded
    depth, each disk intersects at most $O(1/\eps)$ other disks, and
    the intersection graph of $\G = \IGraphX{\DS'}$ has
    $|\EGX{\G}| = O(n/\eps)$ edges, and can be computed in $O(n/\eps)$
    time.  Using the algorithm of Duan and Pettie \cite{dp-ltamw-14}
    on $\IGraphX{\DS'}$, computing a $(1-\eps/2)$-approximate matching
    takes
    $O(|\EGX{\G}|\eps^{-1} \log \eps^{-1}) = O(n\eps^{-2} \log
    \eps^{-1})$ time. It is straightforward to verify that this
    matching, together with the greedy matching of the ``tall'' towers
    yields that desired $(1-\eps)$-approximation.

    As for the running for the case that the diameter $\Delta$ is
    relatively small -- observe that after the cleanup step, there are
    at most $O(\Delta^2)$ towers, and each tower contains $O(1/\eps)$
    disks (each disk intersects $O(1/\eps)$ disks). As such, the
    residual graph has at most $O(\Delta^2/\eps^2)$ edges, and the
    running time of the Duan and Pettie \cite{dp-ltamw-14} algorithm
    is $O\bigl( (\Delta^2/\eps^3)\log(1/\eps) \bigr) $.
\end{proof}

Using a reduction of Bonnet \etal \cite{bcm-mmgig-20}, we can reduce the
residual graph even further, resulting in a slightly faster algorithm.

\begin{lemma}
    \lemlab{m:b:diam:unit:disks:2}%
    Let $\DS$ be a set of $n$ unit disks, and let $\eps \in (0,1)$ be
    a parameter. One can compute, in
    $O\bigl((n/\eps) \log (1/\eps)\bigr)$ time, an $(1-\eps)$-matching
    in $\IGraphX{\DS}$.
\end{lemma}

\NotSoCGVer{%
\begin{proof}
    The algorithm works as follows:
    \begin{compactenumI}
        \smallskip%
        \item It performs tower reduction and computes a greedy
        matching $\CM_1$, as in \lemref{m:b:diam:unit:disks}, by doing
        greedy matching on each tower, till each tower contains
        $O(1/\eps)$ disks. Let $\DS'$ be the set of remaining disks.
        
        \smallskip%
        \item Computes the intersection graph $\G$ of $\DS'$.
        
        \smallskip%
        \item Extracts a subgraph $\GA$ of $\G$ by keeping only a few
        vertices from each tower.  More precisely, the algorithm
        performs the following for each pair of adjacent towers.

        The algorithm computes a greedy matching $\CMC$ between two
        adjacent towers, keeping at most $\lambda+1$ edges in this
        matching. All the edges in $\CMC$ are marked.  Next, the
        algorithm marks up to $\lambda+1$ edges (not in the partial
        matching) from each node in $\CMC$, to vertices in the
        adjacent tower.

        In the end, a disk is kept if it is contained in a marked
        edge. Let $\DS''$ be the resulting set of disks. The
        \emph{surplus} of a tower, defined by a point $\p$, is the
        difference between the number of disks in $\DS'$ that contains
        it, and the number of disks in $\DS''$ that contains it (i.e.,
        the number of disks the tower lost). If the surplus is odd,
        the algorithm adds an arbitrary vertex from the surplus to
        $\DS''$ (reducing the surplus by one, and making it even).
        Let $\GA$ be the intersection graph of $\DS''$.
        
        \smallskip%
        \item $(1-\eps/8)$-approximates a maximum matching $\CM_2$ in
        $\GA$ using the algorithm of Duan and Pettie
        \cite{dp-ltamw-14}.
        
        \smallskip%
        \item Computes a greedy matching $\CM_3$ in $\G - \VX{\CM_2}$,
        by computing a greedy matching for each tower.
        
        \smallskip%
        \item Returns the matching $\CM_1 \cup \CM_2 \cup \CM_3$.
    \end{compactenumI}
    \smallskip%
    Recall that every tower intersects at most $\lambda = O(1)$ other
    towers, see \figref{towers}.  A matching in $\G$, where there is
    at most one edge between two towers, is a \emphi{tower matching}.

    The key observation (due to Bonnet \etal \cite{bcm-mmgig-20}) is
    that for a maximum tower matching in $\G$, for each tower, the
    number of internal (to the tower) matching edges is at least half
    of the surplus value (i.e., all the surplus edges are ``used'' by
    internal matching edges). Thus, given such a tower matching in
    $\G$, one can perform an exchange argument, to get a maximum
    matching of the same size in $\GA$ (after adjusting for the
    surplus matching).  Indeed, one argues that any matching strictly
    between towers in $\G$ can be realized in $\GA$, and then it can
    be supplemented to a maximum matching back in $\G$ by adding
    matching of surplus edges.  See Bonnet \etal \cite{bcm-mmgig-20}
    for details. Thus, it is enough to compute a maximum matching in
    $\GA$ (or approximate it), and then do a greedy matching to get
    the desired approximation.
    
    The graph $\GA$ has $O(n \lambda^2) = O(n)$ edges. Thus, running
    the approximation algorithm of Duan and Pettie \cite{dp-ltamw-14}
    on $\GA$ takes $O\bigl( (n/\eps) \log(1/\eps) \bigr)$ time.
\end{proof}%
}

\begin{remark}
    \remlab{bounded:diam}%
    If the set of disks of $\DS$ has diameter $\Delta$, then the
    cleanup stage reduces the number of disks to
    $O(\Delta^2/\eps)$. Then one can involve the algorithm of
    \lemref{m:b:diam:unit:disks:2}. The resulting algorithm has
    running time $O\bigl(n + (\Delta^2/\eps^2)\log(1/\eps)\bigr)$.
\end{remark}

\begin{remark:unnumbered}
    The above algorithm can be modified to work in similar time for
    shapes of similar size -- the only non-trivial step is computing
    the intersection graph $\GA$. This involves taking all the shapes
    in a tower, and its neighboring towers, computing their
    arrangement, and extracting the intersection pairs. If this
    involves $\nu$ shapes, then this takes $O( \nu^2 )$ time. Each
    shape would be charged $O(\nu^2/\nu)$ amortized time, which
    results in $O(n/\eps)$ time, as $\nu = O(1/\eps)$. The rest of the
    algorithm remains the same.
\end{remark:unnumbered}
\subsection{Matching size estimation}
\seclab{u:estimate}

\begin{theorem}
    \thmlab{u:estimate}%
    Let $\DS$ be a set of $n$ unit disks, and let $\eps \in (0,1)$ be
    a parameter. One can output a number $Z$, such that
    $(1-\eps)\kopt \leq \Ex{Z}$ and
    $\Prob{Z < (1+\eps)\kopt} \geq 1- 1/n^{O(1)}$, where $\kopt$ is
    the size of the maximum matching in $\IGraphX{\DS}$, where
    $\IGraphX{\DS}$ is the intersection graph of $\DS$. The running
    time of the algorithm is
    $O(n + \eps^{-6} \log \eps^{-1} \log^2 n )$.
\end{theorem}

\begin{proof}
    We randomly shift a grid of size length $\psi = \ceil{32/\eps}$
    over the plane, by choosing a random point $\p \in
    [0,\psi]^2$. Formally, the $(i,j)$\th cell in this grid is
    $\p + (i \psi, j \psi) + [0,\psi]^2$, where $i,j$ are
    integers. For a pair of unit disks that intersect, with
    probability $\geq 1-\eps/2$ they both fall into the interior of a
    single grid cell. As such, throwing away all the disks that
    intersect the boundaries of the shifted grid, the remaining set of
    disks $\DS'$, in expectation, has a matching of size at least
    $(1-\eps/8) \kopt$.
    
    For a grid cell $\cell$ in this shifted grid, let $\DS_\cell$ be
    the set of disks of $\DS$ that are fully contained in $\cell$.  A
    grid cell $\cell$ is \emphw{active} if $\DS_\cell$ is not empty.
    Let $\AGC$ be the set of active grid cells. For a cell
    $\cell \in \AGC$, let $\koptiX{\cell}$ be the size of the maximum
    matching in $\IGraphX{\DS_\cell}$.  Using the algorithm of
    \lemref{unit:disk:g}, compute in $O(n)$ time overall, for all
    $\cell \in \AGC$, a number $e_\cell$ such that
    $\koptiX{\cell} /2 \leq e_\cell \leq \koptiX{\cell}$.

    The task at hand is to estimate the sum
    $\sigma = \sum_{\cell \in \AGC} \koptiX{\cell}$, where
    $\sigma \leq \kopt$ and $\Ex{\sigma} \geq (1-\eps/8)\kopt$. To
    this end, we use importance sampling to reduce the number of terms
    in the summation of $\sigma$ that need to be evaluated.  Each term
    $\koptiX{\cell}$ is $1/2$-approximated by $e_\cell$, and thus
    applying the algorithm of \lemref{importance:alg}, to these
    approximation, with $\epsA = \eps/32$, $b=2$, $M=n$, and
    $\BadProb=1/n^{10}$, we get that
    \begin{equation*}
        t =%
        O\pth{ b^4 \epsA^{-2} (\log \log M + \log \BadProb^{-1}) \log
           M }%
        =%
        O( \eps^{-2} \log^2 n )        
    \end{equation*}
    terms need to be evaluated (exactly if possible, but a
    $(1-\eps/16)$-approximation is sufficient) to get $1\pm \eps/8$
    estimate for $\sigma$. For each such cell, we apply the algorithm
    of \remref{bounded:diam}, to get $(1-\eps/16)$-approximation. For
    a cell $\cell$ this takes
    $O\bigl( |\DS_\cell| + (1/\eps^4) \log (1/\eps) \bigr)$
    time. Summing over all these $t$ cells, the running time is
    $O\bigl(n+ (t/\eps^4) \log(1/\eps)\bigr)$.
\end{proof}

\section{Approximate maximum matching for
   general disks}

\subsection{The greedy algorithm}

The following $1/2$-approximation algorithm works (with the same
running time) for any simply connected shapes that are well-behaved.

\begin{lemma}
    \lemlab{greedy:g}%
    Let $\DS$ be a set of $n$ disks in the plane. One can compute a
    greedy matching $\CM$ for $\IGraphX{\DS}$ in $O(n \log n)$
    time. This matching $\CM$ is a $1/2$-approximation -- that is,
    $\cardin{\CM} \geq \kopt/2$, where $\kopt$ is the size of the
    maximum cardinality matching in $\IGraphX{\DS}$.
\end{lemma}

\NotSoCGVer{%
\begin{proof}
    The greedy algorithm repeatedly discovers a pair of disks that
    intersect, add them to the matching, and delete them from
    $\DS$. Naively implemented, this requires quadratic time.
    Instead, we use the standard sweeping algorithm for computing the
    arrangement of circles (i.e., the boundaries of the disks).  The
    algorithm performs the sweeping of the plane by a vertical line
    from left to right. Here, as soon as the sweeping algorithm
    discovers an intersection, it deletes the two disks involved in
    the intersection, and reports this pair as an edge in the
    matching. Naturally, the algorithm deletes the disks from the
    $y$-structure and the sweeping queue. An intersection takes
    $O( \log n)$ time to handle, and there can be at most $O( n)$
    intersections before the algorithm terminates. All other events
    takes $O( n \log n)$ time to handle overall.
\end{proof}
}

\subsection{Approximation algorithm when the graph is sparse}
\seclab{sparse}

\begin{lemma}
    \lemlab{plug:and:play}
    Let $\DS$ be a set of $n$ disks in the plane such that the
    intersection graph $\IGraphX{\DS}$ has $m$ edges. For a parameter
    $\eps \in (0,1)$, one can compute, in
    $O\bigl( n \log n + (m/\eps) \log (1/\eps) \bigr)$ time, an
    $(1-\eps)$-matching in $\IGraphX{\DS}$.
\end{lemma}
\begin{proof}
    Computing the vertical decomposition of the arrangement
    $\ArrX{\DS}$ can be done in $O(n \log n +m )$ randomized time,
    using randomized incremental construction \cite{bcko-cgaa-08}, as
    the complexity of $\ArrX{\DS}$ is $O(n + m)$. This readily
    generates all the edges that arise out of pairs of disks with
    intersecting boundaries.

    The remaining edges are created by one disk being enclosed
    completely inside another disk. One can perform a \DFS on the dual
    graph of this arrangement, such that whenever visiting a
    trapezoid, the traversal maintains the set of disks that contains
    it. This takes time linear in the size of the arrangement, since
    the list of disks containing a point changes by at most one
    element between two adjacent faces. Now, whenever visiting a
    vertical trapezoid that on its non-empty vertical wall on the left
    contains an extreme right endpoint of a disk $\disk$, the
    algorithm reports all the disks that contains this face, as having
    an edge with $\disk$. Since every edge is generated at most $O(1)$
    times by this algorithm, it follows that its overall running time
    is $O(n \log n + m)$.

    Now that we computed the intersection graph, we apply the
    algorithm of Duan and Pettie \cite{dp-ltamw-14}. This takes
    $O\bigl((m/\eps) \log(1/\eps) \bigr)$ time, and computes the
    desired matchings.
\end{proof}

The above is sufficient if the intersection graph is sparse, as is the
case if the graph is low density.

\begin{lemma}
    \lemlab{low:density}
    Let $\DS$ be a set of $n$ disks in the plane with density
    $\cDensity$. For a parameter $\eps \in (0,1)$, one can
    $(1-\eps)$-approximate the maximum matching in $\IGraphX{\DS}$ in
    $O\bigl( n \log n + \frac{n \cDensity}{\eps} \log
    \frac{1}{\eps}\bigr)$ time.
\end{lemma}
\begin{proof}
    The smallest disk in $\DS$ intersects at most $\cDensity$ other
    disks of $\DS$. Removing this disk and repeating this argument,
    implies that $\IGraphX{\DS}$ has at most $\cDensity n$ edges.  The result
    now readily follows from \lemref{plug:and:play}.
\end{proof}

\subsection{The bipartite case}
\seclab{g:bipartite}

Consider computing maximum matching when given two sets of disks
$\DS_1, \DS_2$, where one considers only intersections between disks
that belong to different sets -- that is the bipartite case.  Efrat
\etal \cite{eik-ghbmr-01} showed how to implement one round of
Hopcroft-Karp algorithm using $O(n)$ dynamic range searching
operations on a set of disks.  Using the (recent) data-structure of
Kaplan \etal \cite{kmrss-dpvdg-20}, one can implement this algorithm.
Each operation on the dynamic disks data-structure takes
$O( \log^{11} n )$ time. If our purpose is to get an
$(1-\eps)$-approximation, we need to run this algorithm $O(1/\eps)$
times, so that all paths of length $O(1/\eps)$ get augmented,
resulting in the following.

\begin{lemma}
    \lemlab{bipartite}%
    Given sets $\DS_1, \DS_2$ at most $n$ disks in he plane, one can
    $(1-\eps)$-approximate the maximum matching in the bipartite graph
    \begin{equation*}
        \IGraphX{\DS_1, \DS_2}%
        =%
        (\DS_1 \cup \DS_2, %
        \Set{ \disk_1 \disk_2}%
        {\disk_1 \in \DS_1, \disk_2 \in \DS_2,
           \text{ and } \disk_1\cap \disk_2 \neq \emptyset}).        
    \end{equation*}
    in $O\bigl((n/\eps) \log^{11} n \bigr)$ time. Any augmenting path
    for this matching has length at least $4/\eps$.
\end{lemma}

\subsection{Approximate matching via reduction to %
   the bipartite case}
\seclab{a:g:matching:r:c}

We use a reduction, due to Lotker \etal \cite{lpp-idam-15}, of
approximate general matchings to the bipartite case.

\subsubsection{The Algorithm}
\seclab{colorCodeAlg}

The input is a set $\DS$ of $n$ disks, and a parameter
$\eps \in (0,1)$. The algorithm maintains a matching $\CM$ in
$\IGraphX{\DS}$. Initially, this matching can be the greedy matching.
Now, the algorithm repeats the following $O( \cEps \log n)$ times,
where $\cEps = 2^{8/\eps}$:

\begin{description}
    \item \textbf{$i$\th iteration}: Randomly color the disks of $\DS$
    by two colors (say $1$ and $2$), and let $\DS_1, \DS_2$ be the
    resulting partition. Remove from $\DS_1$ any pair of disks
    $\disk_1, \disk_2$ such that $\disk_1\disk_2$ is in the current
    matching $\CM$. Do the same to $\DS_2$. Let $\CM_i'$ be edges of
    $\CM$ that appear in $\GA_i = \IGraphX{\DS_1, \DS_2}$. Using
    \lemref{bipartite}, find an $(1-\eps/16)$-approximate maximum
    matching in $\GA_i$, and let $\CM_i''$ be this matching. Augment
    $\CM$ with the augmenting paths in $\CM_i' \oplus \CM_i''$.
\end{description}

The intuition behind this algorithm is that this process would compute
all the augmenting paths of $\CM$ of length (say) $\leq 4/\eps$, which
implies that the resulting matching is the desired approximation.

\subsubsection{Analysis}
\begin{lemma}
    \lemlab{aug:rounds}%
    The above algorithm outputs a matching of size
    $\geq (1-\eps)\kopt$, with probability $\geq 1-1/n^{O(1)}$.
\end{lemma}
\NotSoCGVer{%
\begin{proof}
    Let $\MOpt$ denote the maximum cardinality matching in
    $\IGraphX{\DS}$ (thus, $\kopt = |\MOpt|$).  We group the
    iterations of the algorithm into epochs. An \emphi{epoch} is a
    consecutive blocks of $\gamma = 4 \cdot 2^{4/\eps}$ iterations.
    Let $\CM_i$ be the matching computed by the algorithm in the
    beginning of the $i$\th epoch.  Let $d_i = \kopt - |\CM_i|$ be the
    \emphw{deficit} from the optimal solution in the beginning of the
    $i$\th epoch. Initially $d_1 \leq \kopt/2$.  Let $T_i$ be a
    maximum cardinality set of augmenting paths for $\CM_i$, such that
    each path is of length at most
    \begin{equation*}
        \tau = 4/\eps.
    \end{equation*}
    By \lemref{augmenting}, we have $|T_i|\geq (\eps/2)\kopt$.

    Fix a specific path $\pi \in T_i$.  A random coloring by two
    colors, has probability $p \geq 1/2^\tau$ to color $\pi$ such that
    the colors of the disks are alternating. The path $\pi$ is
    \emph{destroyed} in an iteration during the epoch, if the
    algorithm augments along a path that intersects $\pi$.  If $\pi$
    is colored in alternating colors in, then it must have been
    destroyed (in this or earlier iteration) by the algorithm of
    \lemref{bipartite} as it extracts a set of augmenting paths, and
    after it is applied, all remaining augmenting paths have length
    (say) $\geq 8/\eps$ (which is longer then $\pi$).
    
    It follows that after $\gamma$ iterations
    in the epoch, the probability of a path of $T_i$ to survive is at
    most $(1-p)^{\gamma} = (1-p)^{4/p} \leq \exp(-4) \leq 1/50$. Let
    $Z_i$ be the number of paths of $T_i$ that survive the $i$\th
    epoch.  We have that $\Ex{Z_i}\leq |T_i|/50$.  The $i$\th epoch is
    \emphw{successful} if at least half the paths of $T_i$ are
    destroyed in this epoch. By Markov's inequality, we have
    \begin{equation*}
        \Prob{i\text{\th epoch is a failure}}
        =%
        \Prob{ Z_i > |T_i|/2}
        <%
        \frac{\Ex{Z_i}}{|T_i|/2}
        =%
        \frac{1}{25}.
    \end{equation*}
    If the $i$\th epoch is successful, then it computes at least
    $|T_i|/(2\tau) \geq(\eps^2/10)\kopt$ augmenting paths, which
    implies that $d_{i+1} \leq d_i - (\eps^2/10)\kopt$. In particular,
    the algorithm must reach the desired approximation after
    $10/\eps^2$ successful epochs. Since the success of the epochs are
    independent events, it follows, with high probability, that the
    algorithm must collect the desired number of successful epochs,
    after $O( \eps^{-2} + \log n)$ epochs -- this follows readily from
    Chernoff's inequality%
    \NotSoCGVer{, see \lemref{chernoff:long:tail}}.

    Finally, observe that
    $O(4 \cdot 2^{4/\eps} /\eps^2 ) = 2^{8/\eps}$.
\end{proof}
}

\subsubsection{The result}

\begin{theorem}
    \thmlab{a:matching:c:c}%
    Let $\DS$ be a set of $n$ disks in the plane, and $\eps \in (0,1)$
    be a parameter.  One can compute a matching in $\IGraphX{\DS}$ of
    size $\geq (1-\eps)\kopt$, in $O(2^{8/\eps} n \log^{12} n) $ time,
    where $\kopt$ is the cardinality of the maximum matching in
    $\IGraphX{\DS}$.  The algorithm succeeds with high probability.
\end{theorem}

\begin{remark:unnumbered}
    Note, that the above algorithm does not work for fat shapes (even
    of similar size), since the range searching data-structure of
    Kaplan \etal \cite{kmrss-dpvdg-20} can not to be used for such
    shapes.
\end{remark:unnumbered}

\subsection{Algorithm for the case the maximum %
   matching is small}

If $\nCM = \cardin{\CM}$ is small (say, polylogarithmic), it turns out
that one can compute the \emph{maximum} matching exactly in near
linear time.
\begin{lemma}
    \lemlab{preprocess}%
    For a set $\SetA$ of $n$ disks, and any constant
    $\delta \in (0,1)$, one can preprocess $\SetA$, in
    $O(n^{3+\delta} \log n)$ time, such that given a query disk
    $\disk$, the algorithm outputs, in $O( \log n)$ time, a pointer to
    a (unique) list containing all the disks intersecting the query
    disk.
\end{lemma}
\NotSoCGVer{%
\begin{proof}
    Map a disk centered at $(x,y)$ and radius $r$ in $\SetA$ to the
    $45^\circ$ cone in three dimensions with axis parallel to the
    $z$-axis and with an apex at $(x,y, -r)$. Observe that this cone
    intersects the $xy$-plane at a circle that forms the boundary of
    the original disk. A new disk centered at $(x',y')$ of radius $r'$
    intersects the original disk $\iff$ $(x',y',r')$ lies above this
    cone. Thus, every set in $\ranges$ is a face in the arrangement of
    $n$ cones induced by the disks of $\SetA$. This arrangement has
    $O(n^3)$ faces/vertices/edges.  The second result follows by
    preprocessing this arrangement to point location \cite{as-aa-00} -
    this takes $O(n^{3+\delta})$ time, and a point-location query
    takes $O( \log n)$ time, where $\delta \in (0,1)$ is an arbitrary
    fixed constant. 
\end{proof}%
}

\begin{lemma}
    \lemlab{small:matching}%
    Let $\DS$ be a set of $n$ disks in the plane.  Then, in
    $O(n \log n)$ time, one can decide if $\koptX{\DS} =O(n^{1/8})$,
    and if so compute and output this maximum matching.
\end{lemma}

\NotSoCGVer{%
\begin{proof}
    Let $N = cn^{1/8}$, where $c$ is some sufficiently large constant.
    Compute a greedy matching $\CM$ in $\IGraphX{\DS}$ using the
    algorithm of \lemref{greedy:g}. If $\cardin{\CM } > N$, then the
    maximum matching is larger than desired, and the algorithm is
    done.

    Otherwise, let $\VFree = \DS \setminus \VCM$, where $\VCM$ is the
    set of vertices of $\CM$.  The set $\VFree$ is independent, as
    $\CM$ is a greedy matching. Preprocessing the elements of $\VCM$
    using the data-structure of \lemref{preprocess}, we can partition
    the disks of $\VFree$ into $O(N^3)$ classes, where all the disks
    in the same class intersects exactly the same subset of disks of
    $\CM$.  Furthermore, doing point location query for the lifted
    point, corresponding to each disk of $\VFree$, one can decide
    identify its class in $O( \log n)$ time. Observe that all the
    vertices in the same class have exactly the same neighbors in
    $\IGraphX{\DS}$. This takes
    $O(N^{3+\delta} + n \log n) = O(n \log n)$ time.

    Observe that the maximum matching can use at most $2N$ vertices
    that belongs to the same class. Thus, classes that exceed this
    size can be trimmed to this size. Let $\DS'$ be the resulting set
    of disks, and observe that
    $\cardin{\DS'} = O(N^3 \cdot N) = O(N^4)$. Observe that $\DS'$ has
    the same cardinality maximum matching as $\DS$. Furthermore, the
    graph $\G'= \IGraphX{\DS'}$ has at most $m =O(N^5)$ edges.

    Using the maximum cardinality matching algorithm of Gabow and
    Tarjan on $\G'$, takes
    $O(\sqrt{m} | \DS'|) = O(N^{5/2 + 4} ) = o(n)$. We return this as
    the desired maximum matching.
\end{proof}
}

\subsection{Estimation of matching size %
   using separators}
\seclab{recursiveSeparatorApproach}

The input is a set $\DS$ of $n$ disks in the plane with density
$\cDensity$ (if the value of $\cDensity$ is not given, it can be
approximated in near linear time \cite{ah-adrp-08}). Our purpose here
is to $(1-\eps)$-estimate the size of the maximum matching in $\DS$ in
near linear time. Since we can check (and compute it) if the maximum
matching is smaller than $n^{1/8}$ by \lemref{small:matching}, in
$O(n \log n)$ time, assume that the matching is bigger than that.

\subsubsection{Preliminaries}

\begin{lemma}[\cite{h-speps-13}]
    \lemlab{sep:innocent}%
    Let $\p$ be a point in the plane, and let $r$ be a random number
    picked uniformly in an interval $[\alpha, 2\alpha]$. Let $\DSA$ be
    a set of interior disjoint disks in the plane. Then, the expected
    number of disks of $\DSA$ that intersects the circle
    $\Circle = \Circle( \p, r)$, that is centered at $p$ and has
    radius $r$, is $O( \sqrt{\cardin{\DSA}})$.    
\end{lemma}
\NotSoCGVer{%
\begin{proof}%
    We include the proof for the sake of completeness.  For simplicity
    of exposition, translate and scale the plane so that $\p$ is in
    the origin, and $\alpha =1$.  Next, cover the square
    $S = [-3,3]^2$ by a grid of sidelength $\ell = 1 /10$, and let
    $\PS$ be the set of points formed by the vertices of this grid. We
    have $|\PS| = O(1/\ell^2) = O( 1)$. The number of disks of $\DSA$
    that contains points of $\PS$ is bounded by $|\PS|$. Any other
    disk of $\DSA$ that intersects $\Circle$, must be of radius
    $\leq 1/2$, and is fully contained inside $S$. In particular, let
    $D_1, \ldots, D_m$ be these disks, with $r_1, \ldots r_m$ being
    their radii, respectively.  Observe that
    \begin{math}
        \sum_i \pi r_i^2 \leq \mathrm{area}(S) = 36.        
    \end{math}
    and the probability of the $i$\th disk to intersect $\Circle$ is
    at most $2r_i / \alpha = 2r_i$. As such, the expected number of
    disks of $\DSA$ intersecting $\Circle$ is bounded by
    \begin{math}
        |\PS| + \sum_{i=1}^m 2r_i.        
    \end{math}
    However, by the Cauchy-Schwarz inequality, we have
    \begin{math}
        \sum_{i=1}^m r_i%
        \leq%
        \sqrt{\sum_{i=1}^m 1^2} \sqrt{\sum_{i=1}^m r_i^2}%
        \leq%
        \sqrt{m} \sqrt{36 /\pi}%
        =%
        O( \sqrt{\cardin{\DSA}}),
    \end{math}
    as $m \leq \cardin{\DSA}$.
\end{proof}%
}

\subsubsection{Algorithm idea and divisions}

A natural approach to our problem is to break the input set of disks
into small sets, and then estimate the maximum matching size in each
one of them. The problem is that for this to work, we need to
partition the disks participating in the optimal matching, as this
matching can be significantly smaller than the number of input
disks. Since we do not have the optimal matchings, we would use a
proxy to this end -- the greedy matching. The algorithm recursively
partitions it using a random cycle separator provided by
\lemref{w:s:low:density}. We then partition the disks into three sets
-- inside the cycle, intersecting the cycle (i.e., the separator), and
outside the cycle. The algorithm continues this partition recursively
on the in/out sets, forming a partition hierarchy.

\begin{remark}
    \remlab{boundary}%
    For a set generated by this partition, its \emphi{boundary} is the
    set of all disks that intersect it and are not in the set. The
    algorithm maintains the property that for such a set with $t$
    disks, the number of its boundary vertices is bounded by
    $O( \cDensity + \sqrt{\cDensity t})$. This can be ensured by
    alternately separating for cardinality of the set, and for the
    cardinality of the boundary vertices, see \cite{hq-aapel-17} and
    references therein for details. For simplicity of exposition we
    assume this property holds, without going into the low level
    details required to ensure this.
\end{remark}

\subsubsection{The algorithm}

The input is a set $\DS$ of $n$ disks in the plane with density
$\cDensity$, and parameters $\eps \in (0,1)$.  The algorithm computes
the greedy matching, denoted by $\CM$, using \lemref{greedy:g}. If
this matching is smaller than $O(n^{1/8})$, then the algorithm
computes the maximum matching using \lemref{small:matching}, and
returns it.

Otherwise, the algorithm partitions the disks of $\DSA = \VX{\CM}$
recursively using separators, creating a separator hierarchy as
described above.  Conceptually, a subproblem here is a region $R$ in
the plane formed by the union of some faces in an arrangement of
circles (i.e., the separators used in higher level of the
recursion). Assume the algorithm has the sets of disks
$\DSA_{\subseteq R} = \Set{ \disk \in \DSA}{\disk \subseteq R}$ and
$\DS_{\subseteq R} = \Set{ \disk \in \DS}{\disk \subseteq R}$ at
hand. The algorithm computes a separator of $\DSA_{\subseteq R}$,
computes the relevant sets for the children, and continues recursively
on the children.  Thus, for a node $u$ in this recursion tree, there
is a corresponding region $R(u)$, a set of active disks
$\DSA_u = \DSA_{\subseteq R(u)}$, and $\DS_u = \DS_{\subseteq R(u)}$.

The recursion stops the construction in node $u$ if
$|\DSA_{u}| \leq b$, where
\begin{equation*}
    b = \cC \cDensity/\eps^2,
\end{equation*}
and $\cC$ is some sufficiently large constant.  This implies that this
recursion tree has $U = O(\kopt/ b)$ leafs.

If a disk of $\DS$ intersects some separator cycles then it is added
to the set of ``lost'' disks $\LS$.  The hierarchy maps every disk of
$\DS \setminus \LS$ to a leaf. As such, for every leaf $u$ of the
separator tree, there is an associated set $\DS_u$ of disks stored
there. All these leaf sets, together with $\LS$, form a disjoint
partition of $\DS$.

The algorithm now computes for every leaf set a greedy matching, using
\lemref{greedy:g}. Let $e_v$ be the size of this matching. Let $\Xi$
be the set of all leaf nodes. The algorithm next
$(1\pm \eps/4)$-estimates $\sum_{ v \in \Xi} \koptX{\DS_v}$, using
importance sampling, with the estimates
$e_v \leq \koptX{\DS_v} \leq 2e_v$, for all $v$.  Using,
\lemref{importance:alg}, this requires computing
$(1-\eps/8)$-approximate maximum matching for
\begin{equation*}
    t = O\pth{ 2^4 \eps^{-2} (\log \log n + \log n)
           \log n } = O( \eps^{-2} \log^2 n)
\end{equation*}
leafs, this is done using the algorithm of \lemref{small:matching} if
the maximum matching is small compared to the number of disks in this
subproblem, and the algorithm of \lemref{low:density} otherwise. The
algorithm now returns the estimate returned by the algorithm of
\lemref{importance:alg}.

\subsubsection{Analysis}
    
\begin{lemma}
    \lemlab{expect:est}%
    We have
    $\Ex{\sum_{v \in L} \koptX{\DS_v}} \geq (1-\eps/4)\koptX{\DS}$.
\end{lemma}

\NotSoCGVer{%
\begin{proof}
    Let $\cSep \in (0,1)$ be the constant (that does not depend on the
    density) from \lemref{w:s:low:density} -- in two dimensions
    $\cSep \approx 9/10$ \cite{h-speps-13}. The depth of the recursion
    of the above algorithm is
    \begin{math}
        h%
        =%
        O\bigl( 1+ \log_{1/\cSep} ({2\cardin{\CM}}/{b}) \bigr)%
        =%
        O( \log \tfrac{\kopt}{\eps^2} ).
    \end{math}

    The expected
    number of disks of $\DSA$ (and thus edges of $\CM$) destroyed by
    this process is    
    \begin{equation*}
        \sum_{i=0}^{h} 2^i 
        O\pth{\cDensity + \cDensity^{1/2} (2\kopt \cSep^i )^{1/2} }.
    \end{equation*}
    This is a geometric summation dominated by the last term -- that
    is, the total loss is bounded by the loss in the parents of the
    leafs of the recursion multiplied by (say) eight.  Let
    $N = O(\kopt / b)$ be the number of leafs in the recursion
    tree. The loss in the leafs is bounded by
    \begin{equation*}
        \lambda%
        =
        N \cdot 
        O\pth{\cDensity + \sqrt{\cDensity b } }
        =%
        O\pth{ \frac{2|\DSA|\sqrt{\cDensity b  }}{ b} }
        =%
        O\pth{ \frac{\kopt \sqrt{\cDensity}}{ \sqrt{ \cC \cDensity/\eps^2}}}
        =%
        O\pth{ \frac{ \eps }{ \sqrt{ \cC}} \kopt }
        <
        \frac{\eps}{64} \kopt,       
    \end{equation*}
    by making $\cC$ sufficiently large. Specifically, let $\DSA'$ be
    the disks of $\DSA$ that do not intersect any of the circles in
    the separator hierarchy -- we have that the expected loss is
    $\ell = \Ex{\cardin{\DSA \setminus \DSA'}} \leq 8 \lambda \leq
    \eps \kopt / 8$.

    Let $\VOpt = \VX{\MOpt}$, and let $\VFree = \DS \setminus
    \DSA$. As the disks in $\DSA$ are the vertices of a maximum
    matching, it follows that $\VFree$ is an independent set. The
    number of vertices of $\VFree$ that intersects the separating
    circles of the separation hierarchy can be arbitrarily larger than
    $\kopt$. But fortunately, for our purposes, we care only about
    bounding the number of disks of $\VOpt$ intersecting these
    cycles. To this end, observe that the expected loss in
    $\VOpt \cap \DSA$ is bounded by $\ell$. As such, we remain with
    the task of bounding the loss in $\VOpt \cap \VFree$.  So consider
    the separator in the top of the hierarchy.  By
    \lemref{sep:innocent}, the expected number of vertices of
    $\VOpt \cap \VFree$ that intersect it, is bounded by
    $O\bigl( \sqrt{\cardin{\VOpt \cap \VFree}} \bigr) = O( \sqrt{
       \kopt})$. However, this quantity is smaller than the number of
    disks of $\DSA$ being cut by this separator. Repeating
    (essentially) the same calculations as above, we get that the
    total expected number of disks of $\VOpt \cap \VFree$ cut by the
    cycles of the separator hierarchy is also bounded by $\ell$.

    The above claim requires some care -- a subproblem at a node $u$
    has the set of vertices from the greedy matching $k_u =
    |\DSA_u|$. Importantly, the boundary set for $u$ has size
    $t = O( \sqrt{k_u})$, see \remref{boundary}. As such, a maximal
    matching in this subproblem (even if we include all the disks
    intersecting the boundary) is of size at most $k_u + t$. Indeed,
    every one of the boundary disks of the greedy matching might now
    be free to be engaged to a different disk in this subproblem (we
    use here the property that all the disks not in the greedy
    matching are independent). The maximum matching can be at most
    twice the size of the greedy matching, which imply that the
    maximum matching size for $\DS_u$ is bounded by (say)
    $2(k_u + t) \leq 3k_u$, which implies that the above bounding
    argument indeed works.

    Combining the two expected bounds on the size of
    $\LS \cap \VOpt \cap \VFree$ and $\LS \cap \VOpt \cap \DSA$
    implies that in expectation, summed over all the leafs, the
    maximum matching, has at least $\kopt - 2 \ell$ edges, which in
    expectation is $\geq (1-\eps/4)\kopt$. This implies that the
    ``surviving'' maximum matching in the leafs is a good
    approximation to the maximum matching.
\end{proof}
}

\begin{lemma}
    The running time of the above algorithm is
    $O(n \log n + \cDensity^{9}\eps^{-19} \log^2 n)$.
\end{lemma}
\NotSoCGVer{%
\begin{proof}
    The separator hierarchy takes $O(n \log n)$ time to build. So,
    consider the $t =O( \eps^{-2} \log^2 n )$ subproblems. If the
    $i$\th subproblem has $n_i$ disks, and $n_i \geq b^8$, then
    computing the maximum matching for this subproblem takes
    $O( n_i \log n_i)$ time. Otherwise, it takes
    \begin{math}
        O( n_i \log n_i + n_i \cDensity/\eps)%
        =%
        O( n_i \log n_i + b^8 \cDensity /\eps)%
        =%
        O(n_i \log n_i +\cDensity^9/\eps^{17}).
    \end{math}
    Summing over all $t$ subproblems, this takes
    $O( n \log n + t\cDensity^9/\eps^{17} ) = 
    O(n \log n + \cDensity^{9}\eps^{-19} \log^2 n)$.
\end{proof}%
}

\begin{theorem}
    \thmlab{g:d::estimte}%
    Given a set $\DS$ of $n$ disks in the plane with density
    $\cDensity$, and a parameter $\eps \in (0,1)$, one can compute in
    $O(n \log n + \cDensity^{9}\eps^{-19} \log^2 n)$ time, a number
    $Z$, such that
    \begin{math}
        (1-\eps)\kopt \leq \Ex{Z}
    \end{math}
    and
    \begin{math}
        \Prob{ Z > (1+\eps) \kopt} < 1/n^{O(1)},
    \end{math}
    where $\kopt = \koptX{\DS}$ is the size of the maximum matching in
    $\IGraphX{\DS}$.
\end{theorem}
\NotSoCGVer{%
\begin{proof}
    The result follows from the above, adding in the guaranties
    provided by the importance sampling.
\end{proof}%
}

\BibTexMode{%
   \SoCGVer{%
      \bibliographystyle{plainurl}%
   }%
   \NotSoCGVer{%
      \bibliographystyle{alpha}%
   }%
   \bibliography{matchings}
}%
\BibLatexMode{\printbibliography}

\NotSoCGVer{%
\appendix

\section{Chernoff inequality}

The following is a standard version of Chernoff inequality.
\begin{theorem}
    \thmlab{Chernoff:0:1}%
    Let $X_1, \ldots, X_n$ be $n$ independent \emph{coin flips}, such
    that $\Prob{X_i =0} =\Prob{X_i=1} = \frac{1}{2}$, for $i=1,\ldots,
    n$. Let $Y = \sum_{i=1}^n X_i$. Then, for any $\Delta >0$, we have
    \begin{math}
        \Prob{Y \leq {n}/{2} - \Delta } \leq \exp \pth{-2\Delta^2/n}.
    \end{math}
\end{theorem}

\begin{lemma}
    \lemlab{chernoff:long:tail}%
    Let $M > 0$ and $n$ be positive integer parameters.  Consider
    performing $u = 2c\ceil{\ln n} + 4M$ independent experiments,
    where each experiments succeeds with probability $\geq 1/2$. Then,
    with probability $\geq 1/n^{c}$, at least $M$ of these experiments
    succeeded.
\end{lemma}
\begin{proof}
    Let $X_i$ be one of the $i$\th experiment succeeded, and let
    $Y =\sum_i X_i$ -- we assume here that $X_i$ has exactly
    probability $1/2$ to success, and $c>2$ is a prespecified
    constant.  By \thmref{Chernoff:0:1}, we have that the probability
    of failure is
    \begin{align*}
      \Prob{ Y < M }
      &=%
        \Prob{ Y < M }
        =
        \Prob{ Y \leq (c\ceil{\ln n} +2M) - c \ceil{\ln n} - M }
        \leq%
        \exp\pth{ -2 \frac{(c\ceil{\ln n} + M)^2}{ 2c\ceil{\ln n} + 4M}}
      \\&%
      \leq%
      \exp\pth{ -2 \frac{(c\ceil{\ln n} + M)^2}{ 4\pth{c\ceil{\ln n} +
      M}}}
      \leq \frac{1}{n^c}.
    \end{align*}
\end{proof}
}

\end{document}